\documentclass[fleqn,12pt,authoryear,longnamesfirst]{elsarticle}

\usepackage{amsmath,amssymb,amsthm,enumitem,fullpage,mathtools,pgfplots,setspace,subdepth}
\usepackage{tikz,ctable}
\pgfplotsset{compat=1.18} 

\newtheorem{theorem}{Theorem}
\newtheorem{proposition}{Proposition}

\theoremstyle{definition}
\newtheorem{definition}{Definition}

\theoremstyle{remark}
\newtheorem{remark}{Remark}

\usepackage[normalem]{ulem}

\begin{document}

\begin{frontmatter}
    \title{Managing cascading disruptions through optimal liability assignment%
    \tnoteref{t1}}
    \tnotetext[t1]{We are grateful to Lars Peter \O sterdal for insightful suggestions.
    We also thank participants at EURO 2024 (Copenhagen), SING19 (Besan\c con), and seminar attendees at VU Amsterdam and the Chinese University of Hong Kong for comments.
    Gudmundsson and Hougaard gratefully acknowledge financial support from the Independent Research Fund Denmark (grant no. 4260-00050B).} 
    
    \author[1]{Jens Gudmundsson}
    \ead{jg@ifro.ku.dk}
    
    \author[1]{Jens Leth Hougaard}
    \ead{jlh@ifro.ku.dk}
    
    \author[2]{Jay Sethuraman}
    \ead{jay@ieor.columbia.edu}
    
    \affiliation[1]{{Department of Food and Resource Economics, University of Copenhagen}}
    \affiliation[2]{{Department of Industrial Engineering and Operations Research, Columbia University}}
    \begin{abstract}
        Interconnected agents such as firms in a supply chain make simultaneous preparatory investments to increase chances of honouring their respective bilateral agreements.
        Failures cascade:
        if one fails their agreement, then so do all who follow in the chain.
        Thus, later agents' investments turn out to be pointless when there is an earlier failure.
        How losses are shared affects how agents invest to avoid the losses in the first place.
        In this way, a \emph{solution} sets agent liabilities depending on the point of disruption and induces a supermodular investment game.
        We characterize all efficient solutions.
        These have the form that later agents---who are not directly liable for the disruption---still shoulder some of the losses, justified on the premise that they \emph{might have failed} anyway.
        Importantly, we find that such indirect liabilities are necessary to avoid unbounded inefficiencies.
        Finally, we pinpoint one efficient solution with several desirable properties.
    \end{abstract}
\end{frontmatter}

\section{Introduction} \label{SECT:intro}

Sequential systems of interconnected agents appear in many shapes and forms in the modern economy, be it physical or digital (sequential search, innovation, refinement, scheduling, and more).
Due to the very interconnectedness of such systems, they are vulnerable to disruptions with subsequent harmful ripple effects.   
Indeed, these indirect cascading effects may substantially exceed the disruption's direct impact \citep{InoueTodo2019}.
The economic literature abounds in examples of how such externalities if left unaddressed misalign incentives \citep[e.g.,][]{Libecap2014,Barnighausenetal2014}, but also in seminal, concrete proposals on how to correct for them \citep[][among others]{Plott1966,ShapleyShubik1969}.
We follow up on this line of research by proposing how to assign liability and share costs in sequential, interconnected systems.

A potential application of our study is supply chain management.
Here, disruptions pose significant challenges to businesses and result in economic losses both through lost sales and by reducing the company's market value \citep{HendricksSinghal2005,TokuiAl2017}.
Responding to disruptive events requires coordination and cooperation among all parties of the supply chain.%
\footnote{Various ways to manage disruption risks in supply chains are described for instance in \citet{ChopraSodhi2004,Sheffi2005,Tang2006,KleindorferSaad2009}.}
A particular critical issue is the assignment of liability.
For instance, in practice it is often observed that product liability is shared among supply chain partners \citep[][]{FanNiFang2020}, as in the 2007 case of Sanyo and Lenovo who shared the cost of recalling laptop batteries because of a faulty design by Sanyo.%
\footnote{For another example, one of the largest recalls in the automative industry had to do with defective airbags made by the Takata Corporation \citep{delRiego2021,Bijgaart2024}.
These had been supplied to automakers worldwide and the cost of the massive ensuing recalls eventually had to be covered also by the automakers as the supplier filed for bankruptcy.}
\citet{Chaoetal2009} argue that such product recall cost sharing can lead to improved product quality and increased supply chain profits.% 
\footnote{See also \citet{Balachandran2005,Lim2001,ReyniersTapiero1995EJOR,ReyniersTapiero1995MS}.}
Our findings provide further managerial insights into how such practices of coordinated liability assignment can improve supply chain resilience and profitability. 

Although there may be contracts in place to manage liability, optimal design of such contracts is a complex task.
We suggest a framework to optimize how losses from disruptions are assigned.
Specifically, we examine how pre-announced liability agreements can reduce cascading harm by encouraging strategic investments in preparedness (e.g., choosing more reliable production methods).
Such investments lower the risk of disruption and improve the overall resilience of the chain.
In practice, our suggested solution can be implemented through blockchain-based smart contracts \citep{GudmundssonMS2023}.
In the specific case of supply chains, a blockchain-based approach may reduce informational asymmetries between manufacturers and consumers, which can lead to systemic welfare gains \citep{JayI,HongXiao2024}.
Moreover, through efficient tracing of defective products, blockchains may also help reduce supply chain risk by enabling selective product recall \citep{JayII}. 

In our stylized model, agents will in sequence and with random success take on the task of honoring the agreement they have with the next (e.g., to refine and deliver goods as input to the supply chain's next step). 
A failure at any point makes all subsequent agents fail as well.
Interpreted for instance in the context of a major construction project, `failure' might mean a delay at one step that delays the project as a whole with associated re-scheduling costs at all later steps.
Agents influence their success through unobservable investments in preparedness, which are done by all prior to knowing whether anyone will fail. 
In this way, agents later in the chain must account for that their investment will turn out to be pointless when there is a failure earlier on. 
Failures cause monetary losses.
In particular, if agent $i$ causes a disruption, it leads to loss $\ell_i$ to agent $i+1$, subsequently loss $\ell_{i+1}$ to agent $i+2$ (because of agent $i+1$'s forced failure), and so on.
How these losses $\ell_i + \dots + \ell_n$ are shared will affect how agents choose to invest to avoid the losses in the first place.
Our objective is to design a \emph{solution}---that is, a systematic way of sharing losses---to incentivize agents to invest efficiently.

A central question is whether ``downstream'' agents $j > i$ should take on some of the costs even when agent $i$ is mainly responsible for the disruption (we denote this $j$'s ``indirect liability'' to contrast the ``direct liability'' assigned~$i$).
From a normative perspective, if there is a significant risk that $j$ would have failed anyway (had $i$ succeeded), then it seems only fair that $j$ should carry some of the burden.
This is further corroborated from an efficiency perspective.
Specifically, our first main result pertains to solutions that incentivize investments that minimize total systemic costs.
These costs comprise all agents' expected losses and investments.
Theorem~\ref{TH:implementation} shows that
all such ``first-best'' solutions agree on agents' direct liabilities, provides a precise expression for these, and shows that they all feature positive indirect liabilities. 
Interestingly, the direct liabilities present a parallel to the seminal Clarke-Groves mechanisms \citep{Clarke1971,Groves1973}.

Still, Theorem~\ref{TH:implementation} mainly pins down first-best direct liabilities and leaves most indirect liabilities unresolved.
This allows us to largely decouple incentives from distributional concerns.
When designing first-best solutions, as long as direct liabilities are set optimally, we have flexibility in how to set indirect liabilities and in consequence fine-tune the distribution of agent (expected) payoffs. 
For this purpose, we identify a particular first-best solution $\phi^*$ with a simple structure on indirect liabilities.
Theorem~\ref{TH:crossEffect} shows that $\phi^*$ internalizes externalities to the point that even if agent $k$ would get to select $i$'s investment on behalf of~$i$, the efficient investment of agent $i$ would remain optimal to $k$ given efficient investment by the others.
Indeed, Theorem~\ref{TH:crossEffect} shows that this feature, which coordinates all agents around efficient investment, is unique to~$\phi^*$. 

A consequence of Theorem~\ref{TH:implementation} is that indirect liability is required for first-best investments.
Theorem~\ref{TH:POA} explores the other side of the coin, namely the intuitive ``disruptor-pays'' solution that instead systematically assigns all losses to the disruptor.
We find that the potential efficiency loss of this seemingly straightforward solution can be arbitrarily large.
In essence, if losses are skewed towards the end of the chain, then ``disruptor-pays'' would induce the possibly many earlier agents to invest quite heavily to avoid these latter losses;
at some point, total investments actually exceed total potential losses.
From a systemic standpoint, agents would have been better off not investing at all, yet in equilibrium, each has high investment.

\medskip \noindent \emph{Related literature.}
Closest to our work are the papers by \cite{BakshiKleindorfer2009} and \cite{GudmundssonMS2023}. \cite{BakshiKleindorfer2009} study a bargaining game with asymmetric information between two parties:
a retailer and an upstream supplier.
As in our case, the aim is to identify a cost-sharing contract that leads to efficient levels of investments for risk mitigation.
Probabilities of disruption are represented by a specific parameterized functional form (which for the supplier differs according to two potential types: ``safe'' and ``unsafe'').
In absence of agreement, the two parties play a non-cooperative Bayesian game given pre-determined loss shares in case of disruption.
By cooperating, they can induce efficient investments through a suitable cost-sharing arrangement where bargaining gains are shared equally.
In contrast, we have a multi-agent structure, no informational asymmetry, no explicit bargaining, and do not restrict probabilities to specific functional forms.

\cite{GudmundssonMS2023} consider a related model of sequentially triggered losses.
Their approach is mainly normative without a strategic, game-theoretic analysis and focus is on fair liability assignment post-disruption.
In contrast, the analysis of the current paper is set pre-disruption.
That is, we explore how contractual arrangements impact incentives for risk mitigation and proactive measures to prevent disruptions.
\citet{GudmundssonMS2023} assume agents only are aware of their direct relations (that is, agent $i$ knows only of agents $i-1$ and $i+1$).
For this reason, the disruptor is unaware of the full consequences of a disruption, and \citet{GudmundssonMS2023} argue therefore that later losses should be shared between the disruptor and later agents.
In contrast, potential losses are common knowledge in the current paper.
Hence, this argument of limited information no longer has any appeal.
Instead, we propose two novel ideas for why liability should not be put squarely on the first to fail.
The first is based in fairness: 
hypothetically, later agents might have failed anyway and therefore should share some of the later losses.
The second is based in efficiency and evident from Theorem~\ref{TH:POA}:
too much liability assigned the disruptor leads to inefficient overinvestment in preventing failure from happening.

As our model is richer than \citeauthor{GudmundssonMS2023}'s (\citeyear{GudmundssonMS2023}) and adds new levels of agent asymmetry, the solutions we identify are more refined than their ``fixed-fraction rules''.

Finally, by topic our study also relates to the Law and Economics literature on optimal liability assignment, which has been an active area of research at least since \cite{Coase1960}.%
\footnote{Other influential studies include, for instance, \citet{Brown1973,MarchandRussell1973,DiamondMirrlees1975,Green1976,EmonsSobel1991}.}
A typical aim in this literature is to analyze how different types of liability rules, in different economic environments, affect socially optimal resource allocation in terms of accident avoidance.
In particular, the prime concern of liability assignment is to what extent injurer or victim should be held responsible as well as whether and how negligence on both sides should be taken into account \citep[see e.g.,][]{Shavell1980,LandesEtal1987,Shavell2007}.
In comparison, we focus on the cascading effect of losses, giving (non-identical) agents dual roles as victim of the upstream agent and injurer of the downstream agent, along a multi-agent chain structure.
Moreover, we focus on 
strict liability: that is, it is only in the role as direct, or indirect, injurer that an agent is deemed liable.

\medskip \noindent \emph{Outline.}
Next, Section~\ref{SECT:model} introduces the model together with the concept of a solution and provides an axiomatic characterization of a natural class of solutions (Proposition~\ref{PR:characterization}).
Section~\ref{SECT:game} turns to strategic preparatory investments and includes our main findings.
The non-cooperative game is defined and shown to be supermodular (Proposition~\ref{PR:supermodularity}).
Moreover, the notion of first-best solutions is introduced and characterized (Theorem~\ref{TH:implementation}).
Next, our main proposal, the solution $\phi^*$, is defined and shown to possess several desirable features (Proposition~\ref{PR:intersect} and Theorem~\ref{TH:crossEffect}).
Finally, we show that the choice of solution truly matters:
we examine the ``disruptor-pays'' solution that, although very intuitive, risks arbitrarily large efficiency losses (Theorem~\ref{TH:POA}).
We conclude in Section~\ref{SECT:conclusion}.

\section{Preliminaries} \label{SECT:model}

We study a model of sequentially triggered losses inspired by the set-up in \citet{GudmundssonMS2023}.
In their work, the point of departure is after losses have been realized and the problem is to fairly allocate losses (liabilities) to agents in view of the cascading nature of the disruption.
Here, in contrast, we will examine the incentives such liability assignment creates for investing in what broadly can be called ``preparedness'' and hence in preventing losses from occurring in the first place.

\subsection{Model}

There is a set of \textbf{agents} $N = \{1, \dots, n\}$.
Agents are ordered such that each agent $i$ is in relation (e.g., has an agreement or contract) with the next agent, $i+1$, say as part of a supply chain.%
\footnote{Strictly speaking, agent $n$ has a contract with agent $n+1$, who ends the chain.
Consistent with the idea that agents only are liable for later losses, $n+1$ cannot be assigned any liability, and we therefore choose not to include $n+1$ in the model.
In a supply chain, $n+1$ may correspond to end-consumers who receive their ordered products late due to delays in the production chain.}
If agent $i$ fails her agreement with $i+1$, then this causes a loss $\ell_i > 0$ to agent~$i+1$.
Moreover, failures have a cascading effect:
if $i$ fails, then so do all subsequent agents $j > i$, and all losses $\ell_i + \dots + \ell_n$ are incurred.
Let $\ell = (\ell_1, \dots, \ell_n) \in \mathbb{R}^n_{>0}$ be the profile of potential \textbf{losses}.

Each agent $i$ can invest into some \textbf{technology} $p_i \colon \mathbb{R}_{\geq 0} \to [0,1]$ to increase chances to honor their agreements.
With investment $x_i \geq 0$, agent $i$'s probability of honoring their agreement is $p_i(x_i)$ (conditional on all agents prior to $i$ having honored theirs, as otherwise $i$ is certain to fail).
An investment profile, \textbf{profile} for short, is $x = (x_i)_{i \in N} \in \mathbb{R}_{\geq 0}^n \equiv X$.
We assume each $p_i$ is (strictly) increasing, strictly concave, differentiable, and such that $p_i(0) = 0$ and $p_i'(x_i) \to \infty$ as $x_i \to 0$.
Let $p = (p_i)_{i \in N}$ be the technology profile.
A \textbf{problem} is completely described by the pair $\langle \ell, p \rangle$ and fixed throughout.
Intuitively, the assumptions imposed on $\langle \ell, p \rangle$ ensure that we can conveniently limit the analysis to first-order conditions and positive investments.

As such, we can imagine a process where agents first choose their investment to prevent failure, and afterwards either honor, or fail to honor, their contracts according to the resulting probabilities $p_i(x_i)$.
With no investment they are sure to fail. The first agent to fail 
is dubbed the {\it disruptor}.
Thus, the disruptor, say agent $i$, causes a sequence of losses $\ell_i,\dots, \ell_n$.

\subsection{Solutions}

A solution is a rule describing how to allocate losses to agents.
As we do not know from the outset who will fail their agreements (if anyone at all), we systematically apportion losses in all possible events.
Hence, a solution represents a plan for all contingencies and may assign very different liability to agent $i$ if $i$, someone prior to $i$, or someone after $i$ is the first to fail.

Formally, a \textbf{solution} is a function $\phi \colon N \times N \to \mathbb{R}_{\geq 0}$.
In the event that agent $i$ is the disruptor (first to fail), the solution $\phi$ assigns liability $\phi(i,j) \geq 0$ to agent~$j$.
The underlying idea is that agents are only liable for later losses.
Hence, successful agents are not liable, i.e.,  $\phi(i,1) = \dots = \phi(i,i-1) = 0$ for $i > 1$, per definition.
Moreover, all incurred losses must be accounted for (solutions are ``balanced''):
$\sum_{j \geq i} \phi(i,j) = \sum_{j \geq i} \ell_j$.
For instance, a straightforward idea is to assign full liability to the disruptor who triggers the loss chain.
This defines the ``disruptor-pays'' solution $\widehat\phi$ with $\widehat\phi(i,i) = \sum_{j \geq i} \ell_j$ and $\widehat\phi(i,j) = 0$ for $i < j$.

However, for losses that are incurred indirectly through $i$'s failure, we contend that fair assignment is more nuanced.
We know that agents prior to $i$ were successful in honoring their agreements and we know that $i$ failed---but we do not know whether, hypothetically, some agent $j > i$ \emph{would have failed} had $i$ succeeded.
If we have reasons to believe a later loss $\ell_j$ would have arisen irrespective of $i$'s actions, then blame for $\ell_j$ should not be put squarely on the disruptor~$i$.
For this reason, we distinguish ``direct liability'' $\phi(i,i)$---what's assigned agent~$i$ who triggers the disruption---from ``indirect liability'' $\phi(i,j)$---what's assigned agent~$j > i$ on the hypothetical premise that they could have failed regardless.
We turn to two desirable properties of solutions.

\subsection{Desirable properties}

Agents are only liable for ``later'' losses.
Hence, even if losses $\ell_j + \dots + \ell_n$ are incurred, agent $k>j$ should only have to bear some fraction of $\ell_k + \dots + \ell_n$.
That fraction should be higher if $k$ actually failed than if $k$, hypothetically, might have failed.
That is to say, $k$'s indirect liability $\phi(j,k)$ should not exceed $k$'s direct liability $\phi(k,k)$. 

\begin{definition}[Higher direct liability]
    For each $j < k$, $\phi(j,k) \leq \phi(k,k)$.
\end{definition}

Consider now agents $i < j < k$.
The rationale for the indirect liability $\phi(i,k)$ is that agent~$k$ might have failed anyway and therefore should shoulder some of the burden.
Suppose now instead agent $i$ succeeds yet agent $j$ fails.
Then the assessment that ``$k$ might have failed anyway'' still stands.
This suggests that $k$ should be equally indirectly liable regardless who's the first to fail, $i$ or $j$.
A solution with this property is said to satisfy \textit{independent indirect liabilities}.

\begin{definition}[Independent indirect liabilities] \label{DEF:FIL}
    For each $i < j < k$, $\phi(i,k) = \phi(j,k)$.
\end{definition}

These two requirements are trivially satisfied by the ``disruptor-pays'' solution $\widehat\phi$ as well as by many more subtle solutions.
In particular, for a solution $\phi$ with \textit{higher direct liability} and \textit{independent indirect liabilities}, agent $j$'s indirect liability $\phi(i,j)$ must be some fraction $1 - \pi_j$ of their direct liability $\phi(j,j)$.
As direct liabilities turn out to be non-zero, we have $\pi_j \in [0,1]$.
In this way, any solution with the two properties is associated with a weight vector $\pi \in [0,1]^n$.
Let $\phi^\pi$ denote the solution parameterized by weights~$\pi$.
Through balance, we can work out an exact expression of the direct liabilities $\phi^\pi(j,j)$.
This is the larger part of the proof of Proposition~\ref{PR:characterization}.

\begin{proposition} \label{PR:characterization}
    A solution $\phi$ satisfies \textit{higher direct liability} and \textit{independent indirect liabilities} if and only if there is $\pi \in [0,1]^n$ such that, for each $i < j$, $\phi(i,j) = (1 - \pi_j) \cdot \phi(j,j)$ and 
    \[
        \phi(i,i) 
        = \ell_i + \sum_{k > i} \prod_{i < j \leq k} \pi_j \ell_k.
    \]
\end{proposition}

\begin{proof}
    It is immediate that $\phi$ satisfies \textit{higher direct liability} and \textit{independent indirect liabilities} if and only if there is $\pi \in \mathbb{R}^n$ such that $\phi(i,j) = (1 - \pi_j) \cdot \phi(j,j)$ for each $i < j$.
    We first verify that each $\pi_j \in [0,1]$.
    
    First, for contradiction, suppose $\phi(i,i) = 0$ for some agent~$i$.
    By balance, \textit{higher direct liability} and \textit{independent indirect liabilities}, and balance again,
    \[
        \sum_{j \geq i} \ell_j
        = \sum_{j \geq i} \phi(i,j) 
        = \sum_{j \geq i+1} \phi(i,j)
        \leq \sum_{j \geq i+1} \phi(i+1,j)
        = \sum_{j \geq i+1} \ell_j,
    \]
    which is a contradiction as $\ell_i > 0$.
    Hence, each $\phi(j,j) > 0$.
    As $\phi(i,j) = (1 - \pi_j) \cdot \phi(j,j) \geq 0$, $\pi_j \leq 1$.
    By \textit{higher direct liability}, $\phi(i,j) \leq \phi(j,j)$, so $\pi_j \geq 0$.

    Let now $\phi$ satisfy \textit{higher direct liability} and \textit{independent indirect liabilities}, so there is $\pi \in [0,1]^n$ with $\phi(i,j) = (1 - \pi_j) \cdot \phi(j,j)$ for each $i < j$.
    It remains to show that direct liabilities have the desired form.
    By balance, $\phi(n,n) = \ell_n$.
    Let $\pi_{n+1} \equiv 0$ and assume that, for each $j > i$, 
    \[
        \phi(j,j) = \ell_j + \pi_{j+1} \cdot \phi(j+1,j+1).
    \]
    We will show that this holds for $i$ as well.
    By balance and \textit{independent indirect liabilities},
    \begin{align*}
        \phi(i,i) 
        &= \sum_{j \geq i} \ell_j - \sum_{j > i} \phi(i,j)
        = \ell_i + \sum_{j > i} \left ( \ell_j -  (1 - \pi_j) \cdot \phi(j,j) \right ) \\
        &= \ell_i + \sum_{j > i} \left ( \ell_j + \pi_j \cdot \phi(j,j) - (\ell_j + \pi_{j+1} \cdot \phi(j+1,j+1)) \right) \\
        &= \ell_i + \sum_{j > i} \left ( \pi_j \cdot \phi(j,j) - \pi_{j+1} \cdot \phi(j+1,j+1) \right ).
    \end{align*}
    This simplifies to the desired $\phi(i,i) = \ell_i + \pi_{i+1} \cdot \phi(i+1,i+1)$.
    Expanding the recursive expression yields the desired direct liabilities:
    \begin{align*}
        \phi(i,i) 
        &= \ell_i + \pi_{i+1} \cdot \phi(i+1,i+1) 
        = \ell_i + \pi_{i+1} \ell_{i+1} + \pi_{i+1} \cdot \phi(i+2,i+2)
        = \dots \\
        &= \ell_i + \pi_{i+1} \ell_{i+1} + \pi_{i+1} \pi_{i+2} \ell_2 + \dots + \pi_{i+1} \cdots \pi_n \ell_n. \qedhere
    \end{align*}
\end{proof}

\begin{remark}[Three-agent illustration] \label{REM3}
    With three agents, there are three potential points of disruption and Table~\ref{TAB:REM3} summarizes the respective liabilities.
    If agent $3$ is the first (and only) to fail, then $3$ is fully liable for the resulting loss:
    $\phi(3,3) = \ell_3$.
    This is agent $3$'s direct liability.
    If instead agent $2$ is the first to fail, then agent~$3$ is assigned less, namely $\phi(2,3) = (1 - \pi_3) \cdot \phi(3,3) = (1 - \pi_3) \ell_3$.
    This is agent $3$'s indirect liability.
    By balance, the remaining $\ell_2 + \ell_3 - \phi(3,3) = \ell_2 + \pi_3 \ell_3 = \ell(2,2)$ is agent $2$'s direct liability.
    
    \begin{table}[!htb]
        \centering
        \begin{tabular}{cccc} \toprule
            Disruptor $i$ & Agent $j = 1$ & $j = 2$ & $j = 3$ \\ \midrule
            $i = 1$ & $\ell_1 + \pi_2 (\ell_2 + \pi_3 \ell_3)$ & $(1 - \pi_2) ( \ell_2 + \pi_3 \ell_3 )$ & $(1 - \pi_3) \ell_3$ \\
            $i = 2$ & $0$ & $\ell_2 + \pi_3 \ell_3$ & $(1 - \pi_3) \ell_3$ \\
            $i = 3$ & $0$ & $0$ & $\ell_3$ \\
            \bottomrule
        \end{tabular}
        \caption{Liabilities $\phi(i,j)$ for the case of three agents.}
        \label{TAB:REM3}
    \end{table}
    
    In a production context, a pragmatic alternative is to set weights $\pi$ to the fault rates inherent to agents' production.
    For instance, if there is a 10\% risk that agent $3$ fails (if the others succeed), then we could set $\pi_3 = 9/10$ and have $\phi(2,3) = \ell_3 / 10$ and $\phi(2,2) = \ell_2 + (9/10) \ell_3$.
    Finally, if the disruption happens already with agent $1$, then agent $3$'s indirect liability $\phi(1,3) = \phi(2,3)$ is unchanged, agent $2$ pays $\phi(1,2) = (1 - \pi_2) \cdot \phi(2,2)$, and the remaining losses are covered by agent~$1$.
    \hfill \textit{End of remark}
\end{remark}

Next, we turn to analyze the incentives that solutions create for agents to honor their agreements. 
This defines a game in which agents make strategic investments in preventing failure. 

\section{Strategic preparatory investments} \label{SECT:game}

Agents simultaneously invest to affect their chances of honoring their agreements through their individual technology $p_i$: say, by investing in more or less reliable machinery.
Investments $x=(x_1,\dots,x_n)$ are unobservable and done by all prior to knowing whether any agreement will fail.
In this way, agents later in the chain must account for that their investment in preparation will turn out to be pointless if there is a failure earlier on.
Our interest is solely in this investment decision.
For instance, we do not model how agents may exert effort to honor their agreement beyond their preparation.
How losses are shared affects how agents choose to invest.
Our main objective is to devise a solution $\phi$ that incentivizes efficient investments that minimize total system costs (to be defined in Subsection~\ref{SUB:firstbest}).

\subsection{The game}

A solution $\phi$ partly determines agent payoffs and induces a particular game in which agents choose investment levels strategically.
Agents are risk neutral and minimize expected costs.
The expected cost to agent $1$ at profile $x$ under solution $\phi$ is
\[
    C_1(x; \phi) = (1 - p_1(x_1)) \cdot \phi(1,1) + x_1.
\]
For agent $k > 1$, the expected cost is 
\[
    C_k(x; \phi)
    = \sum_{j < k} \prod_{i < j} p_i(x_i) (1 - p_j(x_j)) \cdot \phi(j,k) + \prod_{i < k} p_i(x_i) (1 - p_k(x_k)) \cdot \phi(k,k) + x_k.
\]
The first term is $k$'s potential indirect liability (which $k$ cannot affect, and which agent $1$ does not have a counterpart to), the second is the direct liability, and the third is investment paid upfront.
An immediate observation for agents $i < k$ is that earlier costs $C_i$ are unaffected by later investments~$x_k$. 
In this way, each solution $\phi$ induces a simultaneous-move game between agents $N$, each with strategy space $\mathbb{R}_{\geq 0}$ and payoff function $C_i$ that depends on $\langle \ell, p \rangle$ as well as~$\phi$.
Investments are not observable, but otherwise all information is common knowledge.

\subsection{Supermodularity}

Proposition~\ref{PR:supermodularity} is an observation that holds universally for all solutions~$\phi$ and induced games.
The more earlier agents invest, the likelier that they will honor their agreements and that each later agent may become the disruptor.
In response, later agents will also invest more. 
In this ordered sense, investments are strategic complements and the induced games have several appealing features \citep{Topkis1979,Vives1990,MilgromRoberts1990,HougaardMS2022}. 

\begin{proposition} \label{PR:supermodularity}
    Each solution $\phi$ induces a supermodular, dominance-solvable game with a unique Nash equilibrium.
\end{proposition}

\begin{proof}
    Fix a solution $\phi$ and a profile $x \in X$.
    We will show that, for each agent $i$ and $k$, 
    \[
        \frac{\partial^2 C_k(x; \phi)}{\partial x_i \partial x_k} \leq 0.
    \]
    (That is, $C_k$ is submodular;
    the game is supermodular as $k$ wants to minimize $C_k$.)
    This is immediate for $k < i$ as earlier agents are unaffected by later investments, so let $i < k$.
    For each agent $j$ and investment $x_j$, $p_j(x_j) \geq 0$, $p'_j(x_j) \geq 0$, and $\phi(j,j) \geq 0$.
    Hence,
    \begin{align*}
        \frac{\partial C_k(x; \phi)}{\partial x_k}
        &= 1 - \frac{p_k'(x_k)}{p_k(x_k)} \prod_{j \leq k} p_j(x_j) \cdot \phi(k,k) \\
        \implies
        \frac{\partial^2 C_k(x; \phi)}{\partial x_i \partial x_k}
        &= - \frac{p'_i(x_i)}{p_i(x_i)} \frac{p_k'(x_k)}{p_k(x_k)} \prod_{j \leq k} p_j(x_j) \cdot \phi(k,k)
        \leq 0. 
    \end{align*}

    As $C_1(x; \phi) = (1 - p_1(x_1)) \cdot \phi(1,1) + x_1$ only depends on $x_1$, agent $1$'s optimal investment is independent of the investments of the others, $x_{>1}$.
    Moreover, as $C_1$ is strictly convex in $x_1$, there is a unique optimal investment $\tilde{x}_1$. 

    Next, $C_2(x; \phi)$ similarly is independent of $x_{>2}$.
    Moreover, restricting to agent $1$'s dominant strategy $\tilde{x}_1$, $C_2((\tilde{x}_1,x_{>1}); \phi)$ only varies in $x_2$.
    Again, we find that there is a unique optimal investment $\tilde{x}_2$.
    Repeating the argument, we pin down the profile $\tilde{x}$ through iterated elimination of strictly dominated strategies.
    Moreover, $\tilde{x}$ is also the unique Nash equilibrium.
\end{proof}

In this way, we can speak unambiguously about how solution $\phi$ \emph{implements} profile $x$, as $x$ will be the unique equilibrium profile in the game induced by $\phi$.
Next, we turn to \emph{first-best implementation} and seek solutions $\phi$ to implement efficient investments.

\subsection{First-best implementation} \label{SUB:firstbest}

The \textbf{total cost} $\mathbb{C} \colon X \to \mathbb{R}_{\geq 0}$ comprises all agents' expected losses and investments.
As solutions are balanced, $\mathbb{C}(x) = \sum_i C_i(x; \phi)$ for each solution~$\phi$.
(That is, $\phi$ is ``internal accounting''---it determines who pays what, but that is irrelevant when computing the system's total cost.)
Put differently, for each agent $j$, the loss $\ell_j$ is incurred as long as some agent $i \leq j$ fails.
This happens with probability $1 - p_1(x_1) \cdots p_j(x_j)$. 
Thus, the total cost is 
\[
    \mathbb{C}(x) = \sum_j \big (1 - \prod_{i \leq j} p_i(x_i) \big ) \ell_j + \sum_j x_j.
\]

A profile $x^* \in X$ is \textbf{efficient} if $\mathbb{C}(x^*) \leq \mathbb{C}(x)$ for all $x \in X$.
Next, Proposition~\ref{PR:efficiency} shows that there always is a unique efficient profile and that it entails positive investments. 
For each agent~$i$, the first-order condition is given by 
\[
    \frac{\partial \mathbb{C}(x)}{\partial x_i} = 0
    \iff 
    \frac{p_i'(x_i)}{p_i(x_i)} \sum_{k \geq i} \prod_{j \leq k} p_j(x_j) \ell_k = 1.
\]
Hence, $x^*$ is the solution to the system of $n$ such equations.
For $i < j < k$, agent $j$'s first-order condition depends both on all earlier and all later investments $x_i$ and~$x_k$.

\begin{proposition} \label{PR:efficiency}
    There is a unique efficient profile $x^*$ and $x^*_i > 0$ for each agent~$i$.
\end{proposition}

\begin{proof}
    We proceed in three steps.
    Existence of an efficient profile follows as $\mathbb{C}$ is continuous and we can restrict attention to a compact subset of $X$;
    positive investments hinge on positive losses and $p_i$ steep at zero; 
    uniqueness finally is a consequence of strict concavity of $p_i$. 

    \textsc{Existence:}
    For profiles $x \in X$ with $\sum x_i \geq \sum \ell_i$, we have $\mathbb{C}(x) \geq \sum x_i \geq \sum \ell_i = \mathbb{C}(0,\dots,0)$.
    Hence, for the purpose of minimizing $\mathbb{C}$, it suffices to restrict to profiles in the compact set $\{ x \in X \mid \sum x_i \leq \sum \ell_i \}$.
    As each $p_i$ is differentiable, $\mathbb{C}$ is continuous and we can appeal to the extreme value theorem.
    Hence, there exist efficient profiles.

    \textsc{Positivity:}
    By contradiction, suppose there is an efficient profile $y$ in which some agent invests zero.
    Let $j$ be the first such agent:
    $y_j = 0$ whereas $y_i > 0$ for each $i < j$.
    As $j$ is sure to fail, no later agent invests either:
    we have $y_k = 0$ for $k > j$.
    For each $x_j \geq 0$, 
    \[
        \mathbb{C}(y) - \mathbb{C}(x_j, y_{-j})
        = \prod_{i < j} p_i(y_i) p_j(x_j) \ell_j - x_j.
    \]
    This is zero at $x_j = 0$.
    It suffices to show that the derivative is positive at $x_j = 0$, as then, for sufficiently small but positive $x_j$, $\mathbb{C}(y) > \mathbb{C}(x_j, y_{-j})$, contradicting that $y$ is efficient.
    Differentiate and use that $p_i(y_i) > 0$ for $i < j$, $\ell_j > 0$, and $p_j'(x_j) \to \infty$ as $x_j \to 0$:
    \[
        \frac{\partial }{\partial x_j} \Big ( \prod_{i < j} p_i(y_i) p_j(x_j) \ell_j - x_j \Big )
        = \prod_{i < j} p_i(y_i) p_j'(x_j) \ell_j - 1 \to \infty \text{ as } x_j \to 0.
    \]
    Hence, efficient investments are positive.

    \textsc{Uniqueness:}
    This follows by standard arguments as each $p_i$ is strictly concave \citep[e.g.][]{BoydVandenberghe2004}.
\end{proof}

A \emph{first-best solution} $\phi$ is one that implements efficient investments~$x^*$.
First-best solutions always exist (we give one example in Definition~\ref{DEF:phistar} below), and Theorem~\ref{TH:implementation} shows that all first-best solutions agree on agents' direct liabilities. 
Furthermore, by balance we can conclude that indirect liabilities must be positive for first-best solutions.

\begin{theorem} \label{TH:implementation}
    A solution $\phi$ implements efficient investments $x^*$ if and only if, for each agent~$i$,
    \[
        \phi(i,i) = \ell_i + \sum_{k > i} \prod_{i < j \leq k} p_j(x^*_j) \ell_k.
    \]
\end{theorem} 

\begin{proof}
    As $\mathbb{C}$ is differentiable and minimized at $x^*$, $\partial \mathbb{C}(x^*) / \partial x_i = 0$ for each $i$.
    Similarly, as $C_i$ is differentiable, we must have $\partial C_i(x^*; \phi) / \partial x_i = 0$ to incentivize investment $x^*_i$.
    We have
    \begin{align*}
        \frac{\partial \mathbb{C}(x^*)}{\partial x_i}
        &= 1 - \frac{p_i'(x^*_i)}{p_i(x^*_i)} \sum_{k \geq i} \prod_{j \leq k} p_j(x^*_j) \ell_k \\
        \frac{\partial C_i(x^*; \phi)}{\partial x_i}
        &= 1 - \frac{p_i'(x^*_i)}{p_i(x^*_i)} \prod_{h \leq i} p_h(x^*_h) \cdot \phi(i,i).
    \end{align*}
    These coincide when
    \[
        \sum_{k \geq i} \prod_{j \leq k} p_j(x^*_j) \ell_k =
        \prod_{h \leq i} p_h(x^*_h) \cdot \phi(i,i),
    \]
    which yields
    \[
        \phi(i,i) 
        = \sum_{k \geq i} \prod_{i < j \leq k} p_j(x^*_j) \ell_k
        = \ell_i + \sum_{k > i} \prod_{i < j \leq k} p_j(x^*_j) \ell_k. \qedhere
    \]
\end{proof}

The condition defining first-best direct liabilities can be interpreted as follows.
Suppose all agents prior to $i$ honor their agreements and that all others invest optimally, $x^*_{-i}$. 
Now, look at the difference in total cost if $i$ fails \tikz{\node [draw,circle,inner sep=0,minimum width=8] {\tiny 1};} versus the expected total cost if $i$ succeeds \tikz{\node [draw,circle,inner sep=0,minimum width=8] {\tiny 2};} given that agents $j > i$ invest $x^*_j$ (investments $\sum x_k$ are omitted as these are incurred in both cases):
\[
    \underbrace{\sum_{k \geq i} \ell_k}_{\tikz{\node at (0,0) {}; \node [draw,circle,inner sep=0,minimum width=8,xshift=-5,yshift=-3] {\tiny 1};}}
    - \underbrace{\sum_{k > i} (1 - \prod_{i < j \leq k} p_j(x^*_j) ) \ell_k}_{\tikz{\node at (0,0) {}; \node [draw,circle,inner sep=0,minimum width=8,xshift=-5,yshift=-3] {\tiny 2};}}
    = \ell_i + \sum_{k > i} \prod_{i < j \leq k} p_j(x^*_j) \ell_k.
\]
Hence, $i$'s direct liability is the net harm that $i$'s failure imposes on society conditional on optimal investments by the others.
In this regard, direct liabilities resemble a Clarke-Groves mechanism \citep{Clarke1971,Groves1973}.

\subsection{Internalizing externalities}

Theorem~\ref{TH:implementation} only settles direct liabilities of first-best solutions and leaves open indirect liabilities.
Indeed, a consequence of Theorem~\ref{TH:implementation} is that indirect liabilities are necessary for first-best implementation.
Restricting to two-agent problems, first-best implementation and balance immediately pin down the indirect liabilities as well:
we have $\phi(1,2) = \ell_1 + \ell_2 - \phi(1,1) = (1 - p_1(x^*_1)) \cdot \phi(1,1)$.
Along the same lines, moving to three agents, we can determine $\phi(2,3) = (1 - p_2(x^*_2)) \cdot \phi(2,2)$.
However, we cannot pin down $\phi(1,2)$ and $\phi(1,3)$ as balance only gives a condition on their sum:
$\phi(1,2) + \phi(1,3) = \ell_1 + \ell_2 + \ell_3 - \phi(1,1)$.
Hence, there may be many first-best solutions.
Among these, we define the solution $\phi^*$ with first-best direct liabilities $\phi^*(i,i)$ and independent indirect liabilities $\phi^*(i,j)$. 
\begin{definition}[Solution $\phi^*$] \label{DEF:phistar}
    For each agent $i$ and $j$, $\phi^*(i,j) = (1 - p_i(x^*_i)) \cdot \phi^*(i,i)$ and
    \[
        \phi^*(i,i) = \ell_i + \sum_{k > i} \prod_{i < j \leq k} p_j(x^*_j) \ell_k.
    \]
\end{definition}

\begin{remark}[Two-agent illustration of $\phi^*$]
    Consider the two-agent instance in which agent~$1$ fails to deliver an input to agent~$2$, who consequently cannot deliver the end product to the consumers.
    At a first glance, one could argue that, absent any evidence of wrong-doing by agent~$2$, agent~$1$ should carry all losses.
    But if there is high probability that agent $2$ would have failed regardless of $1$'s actions, it seems appropriate that the loss $\ell_2$ caused to the consumers should be shared between $1$ and~$2$.
    Indeed, a reasonable solution would be to assign agent $2$ the expected loss they create.
    But as investments are unobservable, we cannot compute this expected loss.
    However, losses and technologies are common knowledge and, therefore, efficient investments $x^*$ can be determined.
    The defining feature of $\phi^*$ then is to assign agent $2$ liability that is the expected loss they create \emph{conditional on efficient investments}.
    Remarkably, assigning liabilities in a way that would be fair under efficient investment actually also incentivizes efficient investment.
    \hfill \textit{End of remark}
\end{remark}

Proposition~\ref{PR:intersect} shows that $\phi^*$ is the unique first-best solution that satisfies \textit{independent indirect liabilities}.
In this way, it is the unique intersection of the first-best solutions of Theorem~\ref{TH:implementation} and the class characterized in Proposition~\ref{PR:characterization}. 

\begin{proposition} \label{PR:intersect}
    A solution $\phi$ satisfies \textit{independent indirect liabilities} and implements efficient investments if and only if $\phi = \phi^*$.
\end{proposition}

\begin{proof}
    We omit a detailed proof as it follows the same steps as the second part of Proposition~\ref{PR:characterization}'s proof.
    The key observation is that, with $\pi_{n+1} \equiv 0$, we again have 
    \[
        \phi(j,j) = \ell_j + \pi_{j+1} \cdot \phi(j+1,j+1)
    \]
    with $\pi_i = p_i(x^*_i)$ for each agent $i$ and~$j$.
    To sketch the argument, first-best implementation gives $\phi(n,n)$ and $\phi(n-1,n-1)$;
    balance then gives $\phi(n-1,n)$;
    \textit{independent indirect liabilities} gives $\phi(n-2,n)$;
    first-best implementation gives $\phi(n-2,n-2)$; 
    balance gives $\phi(n-2,n-1)$;
    and so on.
\end{proof}

Proposition~\ref{PR:supermodularity} showed that the game was one of strategic complements:
the more earlier agents invest, the more later agents will invest as well.
A different but related question is how later expected costs are affected by earlier investments.
Agent $i$'s investment has two effects on agent $j$'s expected cost.
Positively, it decreases the likelihood that $i$ fails and that $j$ is indirectly liable for~$\phi(i,j)$.
Negatively, it increases the likelihood that $j$ may become the first to fail and thereby directly liable for~$\phi(j,j)$.
The former is smaller but the latter is heavier discounted.
This suggests a potential balancing point were the two offset.
Theorem~\ref{TH:crossEffect} shows that this holds true in the equilibrium under~$\phi^*$.

This implies a striking feature of this equilibrium:
not only does agent $i$ best respond to $x^*_{-i}$ through $x^*_i$, but even if agent $k$ \emph{would get to select $i$'s investment}, $x^*_i$ would be optimal to $k$ given $x^*_{-i}$.
We are not merely aligning $i$'s incentives with minimizing $\mathbb{C}$, we are even aligning it with minimizing every $C_k$.

\begin{theorem} \label{TH:crossEffect}
    A solution $\phi$ satisfies $\partial C_k(x^*; \phi) / \partial x_i = 0$ for each $i \leq k$ if and only if $\phi = \phi^*$.
\end{theorem}

\begin{proof}
    We show first that $\phi^*$ satisfies the conditions.
    By Theorem~\ref{TH:implementation}, that $x^*$ is an equilibrium under $\phi^*$, this holds for $i = k$.
    Let instead $i < k$.
    Recall that 
    \[
        C_k(x^*; \phi^*)
        = \Big ( 1 - \prod_{j < k} p_j(x^*_j) \Big ) \cdot \phi^*(i,k) + \prod_{j < k} p_j(x^*_j) (1 - p_k(x^*_k)) \cdot \phi^*(k,k) + x^*_k.
    \]
    Differentiating and using that $\phi^*(i,k) = (1 - p_k(x^*_k)) \cdot \phi^*(k,k)$,
    \begin{align*}
        \frac{\partial C_k(x^*; \phi^*)}{\partial x_i}
        &= - \frac{p_i'(x^*_i)}{p_i(x^*_i)} \prod_{j < k} p_j(x^*_j) \cdot \phi^*(i,k) + \frac{p_i'(x^*_i)}{p_i(x^*_i)} \prod_{j < k} p_j(x^*_j) (1 - p_k(x^*_k)) \cdot \phi^*(k,k) \\
        &= \frac{p_i'(x^*_i)}{p_i(x^*_i)} \prod_{j < k} p_j(x^*_j) \cdot \big ( (1 - p_k(x^*_k)) \cdot \phi^*(k,k) - \phi^*(i,k) \big ) = 0. 
    \end{align*}

    We turn to the other direction.
    That is, let $\phi$ be such that $\partial C_k(x^*; \phi) / \partial x_i = 0$ for each $i \leq k$ (this is trivially zero for $i > k$ as well).
    By balance, $\mathbb{C}(x^*) = \sum_i C_i(x^*)$;
    hence,
    \[
        \frac{\partial \mathbb{C}(x^*)}{\partial x_i}
        = \frac{\partial \sum_k C_k(x^*)}{\partial x_i}
        = \frac{\sum_k \partial C_k(x^*)}{\partial x_i}
        = 0.
    \]
    That is to say, $\phi$ implements $x^*$.
    It follows from Theorem~\ref{TH:implementation} that $\phi(k,k) = \phi^*(k,k)$ for each agent $k$. 
    Recall now that
    \[
        C_k(x^*; \phi)
        = \sum_{j < k} \prod_{i < j} p_i(x^*_i) (1 - p_j(x^*_j)) \cdot \phi(j,k) + \prod_{i < k} p_i(x^*_i) (1 - p_k(x^*_k)) \cdot \phi(k,k) + x^*_k.
    \]
    Hence,
    \begin{align*}
        \frac{\partial C_k(x^*; \phi)}{\partial x_i}
        = 
        &- \frac{p'_i(x_i)}{p_i(x_i)} \prod_{h \leq i} p_h(x^*_h) \cdot \phi(i,k) 
        + \frac{p'_i(x_i)}{p_i(x_i)} \sum_{i < j < k} \prod_{h < j} p_h(x^*_h) (1 - p_j(x^*_j)) \cdot \phi(j,k) \\
        &+ \frac{p'_i(x_i)}{p_i(x_i)} \prod_{h < k} p_h(x^*_h) (1 - p_k(x^*_k)) \cdot \phi(k,k).
    \end{align*}

    We have concluded that $\phi(k,k) = \phi^*(k,k)$.
    By definition, $\phi^*(i,k) = (1 - p_k(x^*_k)) \cdot \phi^*(k,k)$.
    
    We can now proceed by induction.
    Assume $\phi(j,k) = \phi^*(j,k) = \phi^*(i,k)$ for each $j > i$.
    We want to show that $\phi(i,k) = \phi^*(i,k)$.
    We have
    \begin{align*}
        \frac{\partial C_k(x^*; \phi)}{\partial x_i}
        = 
        &- \frac{p'_i(x_i)}{p_i(x_i)} \prod_{h \leq i} p_h(x^*_h) \cdot \phi(i,k) 
        + \frac{p'_i(x_i)}{p_i(x_i)} \sum_{i < j < k} \prod_{h < j} p_h(x^*_h) (1 - p_j(x^*_j)) \cdot \phi^*(i,k) \\
        &+ \frac{p'_i(x_i)}{p_i(x_i)} \prod_{h < k} p_h(x^*_h) \cdot \phi^*(i,k).
    \end{align*}
    We hone in on the telescoping series in the second term.
    Define the series $u = (u_j)_{i < j < k}$ and $t = (t_j)_{i < j \leq k}$ through
    \[
        u_j 
        = \prod_{h < j} p_h(x^*_h) (1 - p_j(x^*_j)) 
        = \prod_{h < j} p_h(x^*_h) - \prod_{h < j+1} p_h(x^*_h) 
        = t_j - t_{j+1}.
    \]
    Then
    \begin{align*}
        \sum_{i < j < k} \prod_{h < j} p_h(x^*_h) (1 - p_j(x^*_j))
        &= \sum_{i < j < k} u_j
        = \sum_{i < j < k} ( t_j - t_{j+1} )
        = t_{i+1} - t_k \\
        &= \prod_{h \leq i} p_h(x^*_h) - \prod_{h < k} p_h(x^*_h).
    \end{align*}
    It follows then that
    \[
        \frac{\partial C_k(x^*; \phi)}{\partial x_i}
        = 
        \frac{p'_i(x_i)}{p_i(x_i)} \prod_{h \leq i} p_h(x^*_h) (\phi^*(i,k) - \phi(i,k)).
    \]
    As each $x^*_h$ and $p_h(x^*_h)$ is positive, this is zero for $\phi^*(i,k) = \phi(i,k)$.
\end{proof}

Equilibrium costs have a simple form under~$\phi^*$.
For a solution $\phi$ with \textit{independent indirect liabilities} (such as $\phi^*$) and agents $i < j$, we have
\[
    C_j(x; \phi)
    = \Big ( 1 - \prod_{i < j} p_i(x_i) \Big ) \cdot \phi(i,j) + \prod_{i < j} p_i(x_i) (1 - p_j(x_j)) \cdot \phi(j,j) + x_j.
\]
In particular, under $\phi = \phi^*$, we have $(1 - p_j(x^*_j)) \cdot \phi^*(j,j) = \phi^*(i,j)$ for $i < j$.
Hence, the expression simplifies to $C_j(x^*_j, x_{-j}; \phi^*) = C_j(x^*; \phi^*) = \phi^*(i,j) + x^*_j$.
That is, equilibrium expected costs equal the agent's indirect liability together with their investment.

\begin{remark}[Simulations]
    To get better insights on how $\phi^*$ assigns liabilities and the associated expected equilibrium costs, we conduct a simulation study with $n = 8$ agents.
    Losses $\ell$ are drawn independently and uniformly from $1$ to $100$.
    Agents have access to symmetric technologies with
    \[
        p_i(x_i) = \frac{\sqrt{x_i}}{1 + \sqrt{x_i}}.
    \]
    Efficient investments $x^*$ together with the solution $\phi^*$ are computed and averaged over $10,000$ such randomly generated instances.
    The main findings are presented in Figure~\ref{FIG:sims}.

    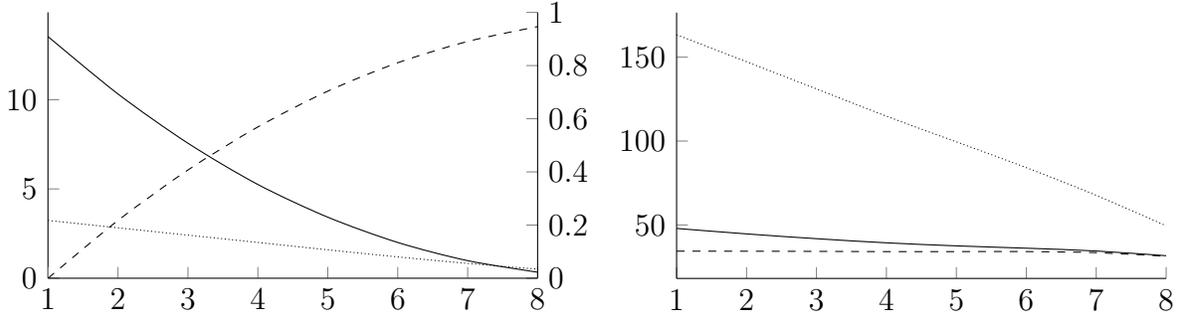
\begin{figure}[!htb]
        \centering
        \begin{tikzpicture}
            \begin{axis}[width=.49\textwidth, height=.31\textwidth, xmin = 1, xmax = 8, ymin = 0, axis y line*=left, hide x axis]
                \addplot[smooth] table [skip first n=1, x index = {0}, y index = {3}, col sep=comma] {data.csv}; % Investment
            \end{axis}
            \begin{axis}[width=.49\textwidth, height=.31\textwidth, xmin = 1, xmax = 8, ymin = 0, ymax = 1, axis y line*=right, axis x line*=bottom]
                \addplot[smooth,densely dotted] table [skip first n=1, x index = {0}, y index = {4}, col sep=comma] {data.csv}; % Probability directly liable
                \addplot[smooth,dashed] table [skip first n=1, x index = {0}, y index = {5}, col sep=comma] {data.csv}; % Probability indirectly liable
            \end{axis}
        \end{tikzpicture}
        \begin{tikzpicture}[] %
            \begin{axis}[width=.49\textwidth, height=.31\textwidth, xmin = 1, xmax = 8, axis y line*=left, axis x line*=bottom]
                \addplot[smooth,densely dotted] table [skip first n=1, x index = {0}, y index = {1}, col sep=comma] {data.csv}; % Direct liability
                \addplot[smooth,dashed] table [skip first n=1, x index = {0}, y index = {2}, col sep=comma] {data.csv}; % Indirect liability
                \addplot[smooth] table [skip first n=1, x index = {0}, y index = {6}, col sep=comma] {data.csv}; % Expected costs
            \end{axis}
        \end{tikzpicture}
        \caption{Left: Investments (solid) and probabilities of becoming directly (dotted, right axis) or indirectly (dashed) liable.
        Right: Expected costs (solid) together with direct (dotted) and indirect (dashed) liabilities.}
        \label{FIG:sims}
    \end{figure}
    
    On the left, we see that investments are decreasing.
    (This is a consequence that holds generally when technologies are equal.)
    Naturally, later agents are more likely to be indirectly liable but also less likely to be directly liable (despite their lower investment).
    On the right, expected costs and direct liabilities are decreasing, although the former is quite shallow.
    Finally, indirect liabilities are surprisingly constant (however, this may be an artefact of our parameter choices).
    Hence, at least here, differences in expected costs are mainly driven by investment differences. 
    \hfill \textit{End of remark}
\end{remark}

\subsection{Efficiency loss}

A way to interpret Theorem~\ref{TH:crossEffect} is that $\phi^*$ induces a pure coordination game in which each agent wants all others to coordinate on $x^*$.
This incentive structure is unique to~$\phi^*$.
For contrast, compare with the ``disruptor-pays'' solution $\widehat\phi$, which assigns full responsibility to the first unsuccessful agent.
As indirect liabilities are zero, later agents want earlier agents to fail.
That is to say, each agent $k$ prefers agent $1$ to invest zero.
However, because of the large direct liability $\widehat\phi(1,1)$, agent $1$ chooses to invest a large amount that exceeds the efficient~$x^*_1$.
But given this, agent $1$ is more likely to succeed and agent $2$ more likely to turn active, so agent $2$ also overinvestments, and so on.
These excess investments can lead to large efficiency losses.
Indeed, Theorem~\ref{TH:POA} shows that the efficiency loss is unbounded. 

With some abuse of notation, let $\phi(\ell, p) \in X$ denote the unique equilibrium profile in the game induced by~$\phi$ given problem $\langle \ell, p \rangle$.
By Proposition~\ref{PR:supermodularity}, this is well-defined.
Note that this is the only part of the analysis in which the problem $\langle \ell, p \rangle$ and the set of agents $N$ is not taken as given.

\begin{theorem} \label{TH:POA}
    For each bound $B > 0$, there is a problem $\langle \ell, p \rangle$ such that
    \[
    \mathbb{C}(\widehat \phi (\ell, p)) > B \cdot \mathbb{C}(\phi^*(\ell,p)).
    \]
\end{theorem}

\begin{proof}
    We first describe the structure of the problems $\langle \ell, p \rangle$ that will be used.
    
    Let $\varepsilon > 0$, $n \in \mathbb{N}$, and $\ell = (\varepsilon, \dots, \varepsilon, 1) \in \mathbb{R}^n_{>0}$.
    Hence, with $\varepsilon$ small, the loss to avoid is $\ell_n = 1$, but doing so requires all to succeed.
    We will construct $p$ such that $\phi^*$ incentivizes essentially zero investments---which will be efficient---whereas $\widehat\phi$ incentivizes essentially unit investment by all. 
    By Theorem~\ref{TH:implementation}, $\phi^*$ is a first-best solution, so $x^* = \phi^*(\ell, p)$ is efficient and $\mathbb{C}(x^*) < \mathbb{C}(0,\dots,0) = \sum \ell_i = 1 + (n-1) \varepsilon$.
    Analogously, let $\widehat x = \widehat\phi(\ell, p)$ be the equilibrium profile in the game induced by $\widehat\phi$ for $\langle \ell, p \rangle$.
    The remainder of the proof will be to construct $p$ such that, for each agent~$i$, $\widehat x_i > 1 - \varepsilon$.
    Then $\mathbb{C}(x^*) < \mathbb{C}(0,\dots,0) \to 1$ as $\varepsilon \to 0$ whereas $\mathbb{C}(\widehat x) > \sum_i \widehat x_i > n (1 - \varepsilon) \to n$ as $\varepsilon \to 0$.
    Hence, for any bound $B > 0$, there is a large-enough population $n > B$ and small-enough $\varepsilon > 0$ for which $\mathbb{C}(\widehat x) > B \cdot \mathbb{C}(x^*)$.

    Now, fix $\varepsilon > 0$, $n \in \mathbb{N}$, and define $0 < \delta < \varepsilon$ such that $(1 - \delta)^{n+1} = 1 - \varepsilon$.
    With some abuse of notation, construct $p_1 = \dots = p_n \equiv p$ such that $p(1 - \varepsilon) = 1 - \delta$ and
    \[
        1 
        < \frac{1}{(1 - \delta)^{n-1}}
        < p'(1 - \varepsilon)
        < \frac{1}{(1 - \delta)^n}.
    \]
    Note here that
    \[
        \frac{1}{(1 - \delta)^n}
        = \frac{1 - \delta}{1 - \varepsilon}
        = \frac{p(1 - \varepsilon)}{1 - \varepsilon}.
    \]
    That is, $p(1 - \varepsilon) > (1 - \varepsilon) \cdot p'(1 - \varepsilon)$, which is necessary for $p$ to be strictly concave.
    In essence, $p$ is obtained by appropriately twisting the linear $\tilde{p}(x) = x$ at $0$ and above $1 - \varepsilon$ to match all conditions assumed on the technology. 

    Agent $1$'s first-order condition is simply $p'(\widehat x_1) \cdot \widehat \phi(1,1) = 1$.
    As $\widehat\phi(1,1) = 1 + (n-1) \varepsilon > 1$, we have $p'(\widehat x_1) < 1 < p'(1-\varepsilon)$.
    As $p$ is concave, $\widehat x_1 > 1 - \varepsilon$.
    We proceed by induction.
    Suppose $\widehat x_i > 1 - \varepsilon$, so $p(\widehat x_i) > 1 - \delta$, for each $i < j$.
    For agent~$j$, the first-order condition is
    \[
        \frac{\partial C_j(\widehat x; \widehat\phi)}{\partial x_j} = 0
        \iff 
        \frac{p'(\widehat x_j)}{p(\widehat x_j)} \prod_{i \leq j} p(\widehat x_i) \cdot \widehat\phi(j,j) = 1.
    \]
    As $\widehat\phi(j,j) = 1 + (n-j) \varepsilon > 1$, we must have
    \[
        1 >
        \frac{p'(\widehat x_j)}{p(\widehat x_j)} \prod_{i \leq j} p(\widehat x_i)
        > p'(\widehat x_j) \cdot (1 - \delta)^{j-1}
        \geq p'(\widehat x_j) \cdot (1 - \delta)^{n-1}
        > \frac{p'(\widehat x_j)}{p'(1 - \varepsilon)},
    \]
    so $p'(\widehat x_j) < p'(1 - \varepsilon)$.
    Again, as $p$ is concave, we find that $\widehat x_j > 1 - \varepsilon$.
    By induction, this holds for each agent~$j$.
    The conclusion now follows along the lines sketched above.
\end{proof}

The main take-away from Theorem~\ref{TH:POA} is that the choice of solution can have a profound effect on total costs.
The straightforward approach to assign only direct liability and no indirect liabilities can be arbitrarily inefficient.
We remark also that \emph{underinvestment} is possible with the same potential for unbounded inefficiencies (e.g., for the solution $\phi$ with $\phi(i,j) = \phi(j,j) = \ell_j$ everywhere).
Moreover, the same proof can be used to obtain a similarly unbounded inefficiency in absolute terms (that is, $\mathbb{C}(\widehat \phi(\ell,p)) > B + \mathbb{C}(\phi^*(\ell,p))$).

\section{Concluding remarks} \label{SECT:conclusion}

We conclude by discussing some alternative interpretations of the model. 

\subsection{Systemic losses}

In the model, the loss $\ell_i$ is interpreted as agent $i$'s harm onto agent $i+1$ from failure to uphold their bilateral agreement.
However, this is but one interpretation.
For instance, the model can also capture sequential damage limitation, where the nature of the damages are more systemic rather than affecting a particular agent.
Imagine a fire starts in one apartment.
Step one may be to prevent the fire from spreading to other buildings, step two to prevent it from spreading within the building, and step three to put out the fire in the apartment.
A failure at any step leads to some systemic loss, label these $L_1 > L_2 > L_3$ accordingly.
We may then associate agent $i$ to the \emph{marginal} loss $\ell_i = L_i - L_{i+1} > 0$.
This is consistent with the supply chain-interpretation where $L_i = \sum_{j \geq i} \ell_j$ are the total losses incurred by $i$'s failure. 

\subsection{Tiered supply chains}

A natural extension is to permit several agents in the same ``tier'' of the supply chain.
We can represent this through a directed acyclic graph where nodes are agents (say, firms) and arcs represent bilateral agreements with associated potential losses.
Figure~\ref{FIG:tiers} illustrates several examples.

\begin{figure}[!htb]
    \centering
    \begin{tikzpicture}[scale=.8]
        \def\x{.57735};
        \draw (0,3) circle (2pt) node (03) {}; 
        \draw (-\x,2) circle (2pt) node (02) {}; 
        \draw (\x,2) circle (2pt) node (12) {}; 
        \draw (-2*\x,1) circle (2pt) node (01) {}; 
        \draw (0,1) circle (2pt) node (11) {}; 
        \draw (2*\x,1) circle (2pt) node (21) {}; 
        \draw (-3*\x,0) circle (2pt) node (00) {}; 
        \draw (-\x,0) circle (2pt) node (10) {}; 
        \draw (\x,0) circle (2pt) node (20) {}; 
        \draw (3*\x,0) circle (2pt) node (30) {}; 
        \path [->] (03) edge (02) edge (12); 
        \path [->] (02) edge (01) edge (11); 
        \path [->] (12) edge (21); 
        \path [->] (01) edge (00) edge (10); 
        \path [->] (11) edge (20); 
        \path [->] (21) edge (30); 
    \end{tikzpicture} \quad \quad
    \begin{tikzpicture}[scale=.8]
        \def\x{.57735};
        \draw (0,3) circle (2pt) node (03) {}; 
        \draw [fill=black!20] (-\x,2) circle (2pt) node (02) {}; 
        \draw [fill=black!20] (\x,2) circle (2pt) node (12) {}; 
        \node at (02) [anchor=south east] {$i$};
        \node at (12) [anchor=south west] {$j$};
        \draw (-2*\x,1) circle (2pt) node (01) {}; 
        \draw [fill=black!20] (0,1) circle (2pt) node (11) {}; 
        \draw (2*\x,1) circle (2pt) node (21) {}; 
        \node at (11) [anchor=north east] {$k$};
        \draw (-3*\x,0) circle (2pt) node (00) {}; 
        \draw (-\x,0) circle (2pt) node (10) {}; 
        \draw (\x,0) circle (2pt) node (20) {}; 
        \draw (3*\x,0) circle (2pt) node (30) {}; 
        \path [->] (03) edge (02) edge (12); 
        \path [->] (02) edge (01) edge (11); 
        \path [->] (12) edge (11) edge (21); 
        \path [->] (01) edge (00) edge (10); 
        \path [->] (11) edge (20); 
        \path [->] (21) edge (30); 
    \end{tikzpicture} \quad \quad
    \begin{tikzpicture}[scale=.8]
        \def\x{.57735};
        \draw [fill=black!20] (0,3) circle (2pt) node (03) {}; 
        \draw (2*\x,3) circle (2pt) node (13) {}; 
        \draw (-\x,2) circle (2pt) node (02) {}; 
        \draw (\x,2) circle (2pt) node (12) {}; 
        \draw (3*\x,2) circle (2pt) node (22) {}; 
        \draw (-2*\x,1) circle (2pt) node (01) {}; 
        \draw (0,1) circle (2pt) node (11) {}; 
        \draw (2*\x,1) circle (2pt) node (21) {}; 
        \draw (4*\x,1) circle (2pt) node (31) {}; 
        \draw (-\x,0) circle (2pt) node (10) {}; 
        \draw (\x,0) circle (2pt) node (20) {}; 
        \draw (3*\x,0) circle (2pt) node (30) {}; 
        \path [->] (03) edge (02) edge (12); 
        \path [->] (13) edge (12) edge (22); 
        \path [->] (02) edge (01) edge (11); 
        \path [->] (12) edge (21); 
        \path [->] (22) edge (21) edge (31); 
        \path [->] (01) edge (10); 
        \path [->] (11) edge (20); 
        \path [->] (21) edge (30); 
    \end{tikzpicture} 
    \caption{Generalized supply chain with multiple agents in the same ``tier''.
    Each arc represents a bilateral agreement with an associated potential loss.}
    \label{FIG:tiers}
\end{figure}
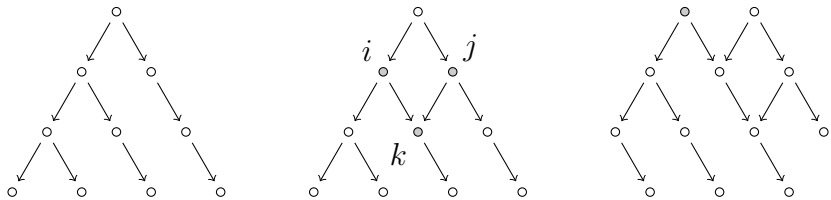

The extension closest to our linear setting is the tree structure illustrated in Figure~\ref{FIG:tiers} (left).
Here, the extended $\phi^*$ would operate in a way completely analogous to the linear case.
For example, say the root node is the disruptor (and, in consequence, each arc that follows fails).
We start with the leaf nodes at the end of the tree, determine their indirect liabilities, and then continue towards the root.
(Alternatively, it may be an arc rather than a node that is the disruptor.
For instance, if the top-left arc fails, then the right-side losses are not incurred.
Still, this only changes the losses to share;
the method is the same.)

A more complex extension is when the graph has a cycle (in the undirected sense). 
Compare how agents $i$ and $j$ both have agreements with agent $k$ in Figure~\ref{FIG:tiers} (middle).
The interpretation closest to ours is that a single failure suffices to trigger the cascade of losses:
that is to say, unless both $i$ and $j$ honor their respective agreements, $k$ will be unable to honor theirs.
It is no longer immediately clear how to define the solution corresponding to $\phi^*$ in the linear case.
In particular, when the root nodes causes the initial failure, it is not clear what shares $i$ and $j$ should carry of the loss that $k$ inflicts on the leaf node that follows. 
However, this ambiguity has to do with indirect liabilities only.
The first-best direct liabilities identified in Theorem~\ref{TH:implementation},
\[
    \phi^*(i,i) = \ell_i + \sum_{k > i} \prod_{i < j \leq k} p_j(x^*_j) \ell_k,
\]
still extend immediately to this setting with the relation ``$i < j$'' interpreted as ``there is a path from $i$ to~$j$ in the graph''.
That is to say, the root node's direct liability that incentivizes efficient investment likely still is completely analogous, but we now have more flexibility in setting indirect liabilities to affect overall payoff distribution.

Finally, Figure~\ref{FIG:tiers} (right) illustrates another difficulty that has to do with timing.
So far, it has been clear from the graph which nodes have failed and which have succeeded.
Now, suppose the top-left node fails.
Whether the top-right node should carry some of the losses clearly depends on whether they have succeeded or not, which we cannot deduce from the graph alone.
Hence, this information has to be added to the model.

\bibliography{bibliography}

\begin{thebibliography}{40}
\providecommand{\natexlab}[1]{#1}
\providecommand{\url}[1]{\texttt{#1}}
\expandafter\ifx\csname urlstyle\endcsname\relax
  \providecommand{\doi}[1]{doi: #1}\else
  \providecommand{\doi}{doi: \begingroup \urlstyle{rm}\Url}\fi

\bibitem[Bakshi and Kleindorfer(2009)]{BakshiKleindorfer2009}
N.~Bakshi and P.~Kleindorfer.
\newblock {Co-opetition and investment for supply chain resilience}.
\newblock \emph{Production and Operations Management}, 18\penalty0
  (6):\penalty0 583--603, 2009.

\bibitem[Balachandran and Radhakrishnan(2005)]{Balachandran2005}
K.~R. Balachandran and S.~Radhakrishnan.
\newblock {Quality Implications of Warranties in a Supply Chain}.
\newblock \emph{Management Science}, 51\penalty0 (8):\penalty0 1266--1277,
  2005.

\bibitem[Boyd and Vandenberghe(2004)]{BoydVandenberghe2004}
S.~Boyd and L.~Vandenberghe.
\newblock \emph{{Convex Optimization}}.
\newblock Cambridge University Press, 2004.

\bibitem[Brown(1973)]{Brown1973}
J.~P. Brown.
\newblock Toward an economic theory of liability.
\newblock \emph{The Journal of Legal Studies}, 2\penalty0 (2):\penalty0
  323--349, 1973.

\bibitem[Bärnighausen et~al.(2014)Bärnighausen, Bloom, Cafiero-Fonseca, and
  O’Brien]{Barnighausenetal2014}
T.~Bärnighausen, D.~E. Bloom, E.~T. Cafiero-Fonseca, and J.~C. O’Brien.
\newblock {Valuing vaccination}.
\newblock \emph{Proceedings of the National Academy of Sciences}, 111\penalty0
  (34):\penalty0 12313--12319, 2014.

\bibitem[Chao et~al.(2009)Chao, Iravani, and Savaskan]{Chaoetal2009}
G.~H. Chao, S.~M.~R. Iravani, and R.~C. Savaskan.
\newblock {Quality Improvement Incentives and Product Recall Cost Sharing
  Contracts}.
\newblock \emph{Management Science}, 55\penalty0 (7):\penalty0 1122--1138,
  2009.

\bibitem[Chopra and Sodhi(2004)]{ChopraSodhi2004}
S.~Chopra and M.~M.~S. Sodhi.
\newblock {Managing risk to avoid: Supply-chain breakdown}.
\newblock \emph{MIT Sloan Management Review}, 46\penalty0 (1):\penalty0 53--61,
  2004.

\bibitem[Clarke(1971)]{Clarke1971}
E.~H. Clarke.
\newblock {Multipart Pricing of Public Goods}.
\newblock \emph{Public Choice}, 11:\penalty0 17--33, 1971.

\bibitem[Coase(1960)]{Coase1960}
R.~H. Coase.
\newblock The problem of social cost.
\newblock \emph{Journal of Law and Economics}, 3:\penalty0 1--44, 1960.

\bibitem[del Riego(2021)]{delRiego2021}
A.~del Riego.
\newblock {Deconstructing Fallacies in Products Liability Law to Provide a
  Remedy for Economic Loss}.
\newblock \emph{American Business Law Journal}, 58\penalty0 (2):\penalty0
  387--447, 2021.

\bibitem[Diamond and Mirrlees(1975)]{DiamondMirrlees1975}
P.~A. Diamond and J.~A. Mirrlees.
\newblock {On the assignment of liability: The uniform case}.
\newblock \emph{The Bell Journal of Economics}, pages 487--516, 1975.

\bibitem[Emons and Sobel(1991)]{EmonsSobel1991}
W.~Emons and J.~Sobel.
\newblock {On the Effectiveness of Liability Rules when Agents are not
  Identical}.
\newblock \emph{Review of Economic Studies}, 58\penalty0 (2):\penalty0
  375--390, 1991.

\bibitem[Fan et~al.(2020)Fan, Ni, and Fang]{FanNiFang2020}
J.~Fan, D.~Ni, and X.~Fang.
\newblock {Liability cost sharing, product quality choice, and coordination in
  two-echelon supply chains}.
\newblock \emph{European Journal of Operational Research}, 284\penalty0
  (2):\penalty0 514--537, 2020.

\bibitem[Green(1976)]{Green1976}
J.~Green.
\newblock On the optimal structure of liability laws.
\newblock \emph{The Bell Journal of Economics}, pages 553--574, 1976.

\bibitem[Groves(1973)]{Groves1973}
T.~Groves.
\newblock {Incentives in Teams}.
\newblock \emph{Econometrica}, 41\penalty0 (4):\penalty0 617--631, 1973.

\bibitem[Gudmundsson et~al.(2024)Gudmundsson, Hougaard, and
  Ko]{GudmundssonMS2023}
J.~Gudmundsson, J.~L. Hougaard, and C.~Y. Ko.
\newblock {Sharing Sequentially Triggered Losses: Automated Conflict Resolution
  Through Smart Contracts}.
\newblock \emph{Management Science}, 70\penalty0 (3):\penalty0 1773--1786,
  2024.

\bibitem[Hendricks and Singhal(2005)]{HendricksSinghal2005}
K.~Hendricks and V.~Singhal.
\newblock {The effects of supply chain glitches on shareholder wealth}.
\newblock \emph{Journal of Operation Management}, 21\penalty0 (5):\penalty0
  501--523, 2005.

\bibitem[Hong and Xiao(2024)]{HongXiao2024}
Z.~Hong and K.~Xiao.
\newblock {Digital economy structuring for sustainable development: the role of
  blockchain and artificial intelligence in improving supply chain and reducing
  negative environmental impacts}.
\newblock \emph{Scientific Reports}, 14\penalty0 (3912), 2024.

\bibitem[Hougaard et~al.(2022)Hougaard, Moreno-Ternero, and
  {\O}sterdal]{HougaardMS2022}
J.~L. Hougaard, J.~D. Moreno-Ternero, and L.~P. {\O}sterdal.
\newblock {Optimal management of evolving hierarchies}.
\newblock \emph{Management Science}, 68\penalty0 (8):\penalty0 6024--6038,
  2022.

\bibitem[Inoue and Todo(2019)]{InoueTodo2019}
H.~Inoue and Y.~Todo.
\newblock {Firm-level propagation of shocks through supply-chain networks}.
\newblock \emph{Nature Sustainability}, 2:\penalty0 841--847, 2019.

\bibitem[Iyengar et~al.(2023)Iyengar, Saleh, Sethuraman, and Wang]{JayI}
G.~Iyengar, F.~Saleh, J.~Sethuraman, and W.~Wang.
\newblock {Economics of permissioned blockchain adoption}.
\newblock \emph{Management Science}, 69\penalty0 (6):\penalty0 3415--3436,
  2023.

\bibitem[Iyengar et~al.(2024)Iyengar, Saleh, Sethuraman, and Wang]{JayII}
G.~Iyengar, F.~Saleh, J.~Sethuraman, and W.~Wang.
\newblock {Blockchain adoptation in a supply chain with manufacturer market
  power}.
\newblock \emph{Management Science}, 2024.

\bibitem[Kleindorfer and Saad(2009)]{KleindorferSaad2009}
P.~Kleindorfer and G.~H. Saad.
\newblock {Managing disruption risks in supply chains}.
\newblock \emph{Production and Operations Management}, 14\penalty0
  (1):\penalty0 53--68, 2009.

\bibitem[Landes and Posner(1987)]{LandesEtal1987}
W.~M. Landes and R.~A. Posner.
\newblock \emph{{The economic structure of tort law}}.
\newblock Harvard University Press, 1987.

\bibitem[Libecap(2014)]{Libecap2014}
G.~D. Libecap.
\newblock {Addressing Global Environmental Externalities: Transaction Costs
  Considerations}.
\newblock \emph{Journal of Economic Literature}, 52\penalty0 (2):\penalty0
  424--479, 2014.

\bibitem[Lim(2001)]{Lim2001}
W.~S. Lim.
\newblock {Producer-Supplier Contracts with Incomplete Information}.
\newblock \emph{Management Science}, 47\penalty0 (5):\penalty0 709--715, 2001.

\bibitem[Marchand and Russell(1973)]{MarchandRussell1973}
J.~R. Marchand and K.~P. Russell.
\newblock Externalities, liability, separability, and resource allocation.
\newblock \emph{The American Economic Review}, 63\penalty0 (4):\penalty0
  611--620, 1973.

\bibitem[Milgrom and Roberts(1990)]{MilgromRoberts1990}
P.~Milgrom and J.~Roberts.
\newblock {Rationalizability, Learning, and Equilibrium in Games with Strategic
  Complementarities}.
\newblock \emph{Econometrica}, 58\penalty0 (6):\penalty0 1255--1277, 1990.

\bibitem[Plott(1966)]{Plott1966}
C.~R. Plott.
\newblock Externalities and corrective taxes.
\newblock \emph{Econometrica}, 33:\penalty0 84--87, 1966.

\bibitem[Reyniers and Tapiero(1995{\natexlab{a}})]{ReyniersTapiero1995EJOR}
D.~J. Reyniers and C.~S. Tapiero.
\newblock {Contract design and the control of quality in a conflictual
  environment}.
\newblock \emph{European Journal of Operational Research}, 82:\penalty0
  373--382, 1995{\natexlab{a}}.

\bibitem[Reyniers and Tapiero(1995{\natexlab{b}})]{ReyniersTapiero1995MS}
D.~J. Reyniers and C.~S. Tapiero.
\newblock {The Delivery and Control of Quality in Supplier-Producer Contracts}.
\newblock \emph{Management Science}, 41\penalty0 (10):\penalty0 1581--1589,
  1995{\natexlab{b}}.

\bibitem[Shapley and Shubik(1969)]{ShapleyShubik1969}
L.~S. Shapley and M.~Shubik.
\newblock {On the Core of an Economic System with Externalities}.
\newblock \emph{The American Economic Review}, 59\penalty0 (4):\penalty0
  676--684, 1969.

\bibitem[Shavell(1980)]{Shavell1980}
S.~Shavell.
\newblock {Strict Liability versus Negligence}.
\newblock \emph{Journal of Legal Studies}, 9\penalty0 (1):\penalty0 1--25,
  1980.

\bibitem[Shavell(2007)]{Shavell2007}
S.~Shavell.
\newblock \emph{{Economic Analysis of Accident Law}}.
\newblock Harvard University Press, 2007.

\bibitem[Sheffi(2005)]{Sheffi2005}
Y.~Sheffi.
\newblock \emph{{The Resilient Enterprise: Overcoming Vulnerability for
  Competitive Advantage}}.
\newblock The MIT Press, 2005.

\bibitem[Tang(2006)]{Tang2006}
C.~S. Tang.
\newblock {Robust strategies for mitigating supply chain disruptions}.
\newblock \emph{International Journal of Logistics Research and Applications},
  9\penalty0 (1):\penalty0 33--45, 2006.

\bibitem[Tokui et~al.(2017)Tokui, Kawasaki, and Miyagawa]{TokuiAl2017}
J.~Tokui, K.~Kawasaki, and T.~Miyagawa.
\newblock {The economic impact of supply chain disruptions from the great
  East-Japan earthquake}.
\newblock \emph{Japan and the World Economy}, 41:\penalty0 59--70, 2017.

\bibitem[Topkis(1979)]{Topkis1979}
D.~M. Topkis.
\newblock {Equilibrium Points in Non-Zero Sum n-Person Sub- modular Games}.
\newblock \emph{SIAM Journal of Control and Optimization}, 17\penalty0
  (6):\penalty0 773--787, 1979.

\bibitem[van~den Bijgaart and Cerruti(2024)]{Bijgaart2024}
I.~van~den Bijgaart and D.~Cerruti.
\newblock {The Effect of Information on Market Activity: Evidence from Vehicle
  Recalls}.
\newblock \emph{The Review of Economics and Statistics}, 106\penalty0
  (1):\penalty0 230--245, 2024.

\bibitem[Vives(1990)]{Vives1990}
X.~Vives.
\newblock {Nash equilibrium with strategic complementarities}.
\newblock \emph{Journal of Mathematical Economics}, 19\penalty0 (3):\penalty0
  305--321, 1990.

\end{thebibliography}
\bibliographystyle{abbrvnat}

\end{document}